\documentclass{iopjournal}

\usepackage{amsmath}
\usepackage{amssymb}
\usepackage{graphicx}%
\usepackage{multirow}%
\usepackage{amsmath,amssymb,amsfonts}%
\usepackage{amsthm}%
\usepackage{mathrsfs}%
\usepackage[title]{appendix}%
\usepackage{xcolor}%
\usepackage{textcomp}%
\usepackage{manyfoot}%
\usepackage{booktabs}%
\usepackage{tikz}
\usepackage{algorithm}%
\usepackage{algorithmicx}%
\usepackage{algpseudocode}%
\usepackage{listings}%
\usepackage{physics}
\usepackage{comment}
\usepackage{graphicx}
\usepackage{subcaption}

\theoremstyle{thmstyleone}%
\newtheorem{theorem}{Theorem}
\newtheorem{proposition}[theorem]{Proposition}
\theoremstyle{thmstyletwo}%

\theoremstyle{thmstylethree}%
\newtheorem{definition}{Definition}%

\begin{document}
\makeatletter
\pagestyle{plain}
\thispagestyle{plain}
\makeatother


\title{Quasi-Orthogonal Stabilizer Design for Efficient Quantum Error Suppression}

\author{Valentine Nyirahafashimana$^{1,3}$\orcid{0000-0002-9849-8163}, Sharifah Kartini Said Husain$^{1,4}$\orcid{0000-0001-5675-941X}, Umair Abdul Halim$^5$\orcid{0000-0003-3113-993X}, Ahmed Jellal$^{6,7}$\orcid{0000-0003-0440-4988} and Nurisya Mohd Shah$^{{1,2,^*}}$\orcid{0000-0001-7783-1118}}

\affil{$^1$Institute for Mathematical Research, Universiti Putra Malaysia, 43400 UPM Serdang, Selangor Darul Ehsan, Malaysia.}\\

\affil{$^2$Department of Physics, Faculty of Science, Universiti Putra Malaysia, 43400 UPM Serdang, Selangor Darul Ehsan, Malaysia}\\

\affil{$^3$ULK Polytechnic Institute, Kigali Campus, Kigali Independent University (ULK), 102 KG 14
Ave Gisozi, Gisozi, 00000, Gisozi, Rwanda.}\\

\affil{$^4$Department of Mathematics and Statistics, Faculty of Science, Universiti Putra Malaysia, 43400
UPM Serdang, Selangor Darul Ehsan, Malaysia.}\\

\affil{$^5$Centre of Foundation Studies in Science, Universiti Putra Malaysia, 43400 UPM Serdang, Selangor Darul Ehsan, Malaysia}\\

\affil{$^6$Laboratory of Theoretical Physics, Faculty of Sciences, Chouaıb Doukkali University, PO Box
20, Serdang, 24000, El Jadida, Morocco}\\

\affil{$^7$Canadian Quantum Research Center, 204-3002 32 Ave Vernon, BC V1T 2L7, Canada.}\\

\affil{$^*$Author to whom any correspondence should be addressed.}

\email{risya@upm.edu.my}


\begin{abstract}

Orthogonal geometric constructions are the basis of many quantum error-correcting codes (QEC), but strict orthogonality constraints limit design flexibility and resource efficiency. We introduce a quasi-orthogonal geometric framework for stabilizer codes that relaxes these constraints while preserving the symplectic commutation structure on the binary symplectic space  \(\mathbb{F}_2^{2n}\). The approach permits controlled overlap between $X$- and $Z$-check supports, leading to quasi-orthogonal Pauli operators and a generalized notion of effective distance defined via induced anti-commutation with logical operators. This relaxation expands the stabilizer design space, enabling codes that approach the Gilbert–Varshamov regime with improved logical rates at moderate distances. Finite-length constructions, including quasi-orthogonal variants of the $[[8,3,\approx 3]]$, $[[10,4,\approx 3]]$, $[[13,1,5]]$, and $[[29,1,11]]$ codes, demonstrate consistent improvements over strictly orthogonal counterparts. Under depolarizing noise with error rates up to $p = 0.30$, logical error rates, fidelities, and trace distances improve by up to two orders of magnitude. These improvements reflect the increased connectivity of the underlying stabilizer geometry while remaining compatible with standard decoding schemes. The proposed framework offers a principled extension of stabilizer code design through quasi-orthogonal geometric structures.

\end{abstract}

\section{Introduction}

\label{sec1}

Quantum error-correcting codes (QECCs) benefit from a quasi-geometric structure that ties design choices to performance. Previous work shows that the genus of the topological and homological invariants, the Euler characteristic, and the nontrivial cycles govern key parameters such as the code distance and the encoding rate. Calderbank et al.~\cite{Calderbank1997quantum} introduced a group-theoretic framework to construct QECCs based on orthogonal geometry, along with the Gilbert–Varshamov lower bound for their asymptotic rate.
Later, Ball et al.~\cite{Ball2023Quantum} mapped code parameters to projective-geometric constructions and analyzed stabilizer and non-stabilizer codes in higher dimensions. It indicated that distance and rate are constrained by geometric and topological invariants. Mahmoud \textit{et al.}~\cite{Mahmoud2025systematic} 
highlighted the potential of hyperbolic geometry for scalable quantum codes and motivated further work on quasi-geometric, non-orthogonal constructions.
Sharma and Papadias~\cite{Sharma2002improve} introduced quasi-orthogonal space-time block codes (QOSTBCs) to overcome interference limitations in joint decoding. Their method applies an optimized constellation rotation $\phi$ to restore the distance properties, yielding notable gains (e.g., $6$\,dB for 4-PSK). Although developed for classical modulation schemes, the geometric rotation principle closely parallels quasi-geometric transformations in quantum settings. 
 Kasai et al.~\cite{Kasai2011Non} proposed a method that creates non-binary quasi-cyclic quantum low-density parity-check (LDPC) codes for two-level QECCs over an extended field of order $2p$, where $p$ is an integer greater than 1. The binary quasi-cyclic LDPC codes (QC-LDPCCs) served as the foundation for the recommended QECCs, which have nearly reached the bounded-distance decoding limit, but have not yet crossed it. 
Cao et al.~\cite{Cao2020quantum} extended matrix-product code (MPC) techniques to construct stronger QECCs using quasi-orthogonal and quasi-unitary matrices (QUMs) over finite fields. Their algebraic framework, which combined CSS/Hermitian constructions with diagonalization in quadratic form , produced dual-containing $q$-ary codes with improved length and distance. Although large QUMs are difficult to build for $q=2$, small examples -$k$  showed clear performance gains. These quasi-geometric matrices act like approximate rotations in a code space.

\vspace{6pt}
Subsequently, Nyirahafashimana \emph{et al.}~\cite{Nyirahafashimana2025Exploring} investigated the role of cycle structures, including six-length cycles and girth, in quasi-cyclic, quasi-orthogonal, and quasi-structured code designs. Their analysis reported reduced computational complexity and improved energy efficiency, together with quantum error rates on the order of \(10^{-5}\) at 10~dB and \(10^{-6}\) at 15~dB, as well as a change in bit-error-rate of approximately 1.20~dB at \(10^{-4}\). These results indicate that code performance is highly sensitive to both cycle distribution and quasi-orthogonality, aligning with trends observed in topological quantum codes. In \cite{Kasami1974gilbert}, the author showed that it is possible to construct long quasi-cyclic binary codes with rate $1/2$ that approach the Gilbert-Varshamov bound (GVB), a key benchmark in error-correcting code performance. By extending earlier results, it developed these codes using polynomial techniques over finite fields and proved that their minimum distance remains strong as the code length increases. In essence, it also demonstrated that structured codes, which are easier to implement, can still achieve near-optimal performance. In~\cite{vu2005improving}, the authors strengthened the GVB for $q$-ary codes by proving that larger code sizes can be achieved for a wide range of parameters. They used a graph-theoretic approach to show that improved bounds are possible without relying on computational methods.

\vspace{6pt}

This work strengthens quantum classical integration by developing quasi-orthogonal geometric methods that overcome the limitations of strict orthogonality in existing QECCs.  Explore how better qubit mapping, scalable code distances, and stable asymptotic rates can be obtained by permitting controlled relaxations in geometric constraints. Additionally, the analysis looks at how error-correction performance is affected under depolarizing and related noise models by quasi-geometric features like code layout, boundary structure, and scaling behavior.

The paper is organized as follows. Section~\ref{Section II} 
presents the mathematical basis of quasi-orthogonal geometry, including singular and non-singular by columns NSC MPCs, quasi-orthogonal Pauli operators, codeword structure, error conditions of qubit mapping, GVB and  performance metrics. Section~\ref{Section III} evaluates quasi-geometric mappings under depolarizing noise through performance metrics analysis. Section~\ref{Section IV} provides the conclusion and outlook

\section{Theoretical Framework}\label{Section II}

This section develops the mathematical basis for quasi-orthogonal qubit geometry, covering NSC MPCs, quasi-orthogonal Pauli operators, controlled codeword structure, and logical state normalization, supported by several mapping examples. It concludes with performance metrics and a summary of the GVB in terms of distance and effective error per qubit.

\subsection{Quasi-Orthogonal Subspace}
\label{subsec:quasi-ortho-geom}

Let \( \mathbb{Z}_2 = \{0, 1\} \) and define \( V = \mathbb{Z}_2^n \) as a binary vector space of dimension \( n \). Each element of \( V \) labels the standard basis of \( \mathbb{R}^{2^n} \), written as \( \ket{v} \), where \( v \in V \).
The geometric foundation of qubit mapping under relaxed orthogonality is enriched by insights from classical space-time coding \cite{Sharma2002improve}. For a class of quasi-orthogonal codes with constant interference, they derived the equivalent channel for jointly decoded symbol pairs $(s_1, s_3)$ as
\begin{equation}
\label{eq:sharma-equiv-channel}
X = \sqrt{\frac{\rho}{M}} G(s_1, 0, s_3, 0) H + V,
\end{equation}
where $G(s_1, 0, s_3, 0)$ is a transmission matrix $4 \times 4$, $H \in \mathbb{C}^{N\times N}$ is the channel, and $V$ is complex Gaussian noise. Crucially, the interference from the orthogonal pair $(s_2, s_4)$ vanishes due to the following
\begin{equation}
\begin{aligned}
G(s_1,0,s_3,0)^H G(0,s_2,0,s_4)+ G(0,s_2,0,s_4)^H G(s_1,0,s_3,0) = 0, 
\label{eq: GGG}
\end{aligned}
\end{equation}
yields decoupled maximum likelihood metrics with a constant signal-to-interference ratio across all channel realizations. In this quasi-geometric framework, qubits are embedded into a three-level qutrit space that preserves geometric relations while providing an additional controlled dimension for richer logical representation. Earlier studies, including the work of Roga et al.~\cite{Roga2010davies}, demonstrated that qubit-level dissipative structures can be extended to qutrit systems, enabling improved  error discrimination and greater flexibility in state transformations within multi-level quantum codes. 
This structure inspires a quasi-geometric qubit mapping: let logical qutrits $\{|0\rangle_L, |1\rangle_L, |2\rangle_L\}$ be encoded into two rotated quasi-orthogonal subspaces $\mathcal{S}_1, \mathcal{S}_3 \subset \mathbb{C}^4$ via
\begin{equation}
|\psi_L\rangle = \sum_{i=0,2} c_i \, U(\phi_i) \, |i\rangle_L, \quad U(\phi_i) = e^{i \phi_i Z},
\end{equation}
where $Z$ is a diagonal phase gate and $\phi_i$ is optimized to maximize the minimum distance in the joint constellation $\mathcal{S}_1 \oplus \mathcal{S}_3$  exactly. The resulting code achieves the full rate with relaxed orthogonality, mirroring the trade-off in Eq~\eqref{eq:sharma-equiv-channel}. In the stabilizer formalism, this corresponds to a diagonal scaling $D = \text{diag}(e^{i\phi_1}, \dots, e^{i\phi_3})$ in the quasi-orthogonal map $\Phi$, ensuring $\bar{S}_A$ remains totally singular. At the same time, improve error separation in the Bloch hyper-sphere. In Fig.~\ref{fig:923-quasi-geo-bloch}, the red plane $\bar{S}_A$ corresponds to the support of $G(s_1,0,s_3,0)$ in Eq.\eqref{eq:sharma-equiv-channel}. Rotation by $\phi$  maximizes the separation of logical states, analogously to the rotation constellation  in \cite{Sharma2002improve}.

\subsection{Quasi-Orthogonal Stabilizer Codes from Singular Subspaces.}

Classical matrix-product codes (MPCs) with non-singular columns (NSC) defining matrices $A\in\mathbb{F}_q,{k\times s}$ are known to deliver long and high-distance linear codes.  
  
\begin{definition} Quasi-orthogonal / quasi-unitary matrices.
A matrix $A\in\mathbb{F}_q,{k\times s}$ is quasi-orthogonal for $k\leq s$, if \(A A^{\!T}= \operatorname{diag}(\lambda_1,\dots,\lambda_k),~\lambda_i\neq 0. \)
Over $\mathbb{F}_{q^2}$ it is quasi-unitary, if $A A^{\!\dagger}$ is diagonal with non-zero entries.
\end{definition}
These structures replace the rigid $A A^{T}=I$ (or $A A^{\dagger}=I$) of orthogonal/unitary designs while preserving the NSC property, producing  dual-containing Euclidean/Hermitian MPCs and, via CSS or Hermitian constructions, quantum codes.
\begin{equation}
\bigl[[sn,\;\sum_{i=1}^{k}t_i-sn,\;\ge\min_{i}(s+1-i)d_i]\bigr]_q.
\end{equation}
even when constituent codes are not fully orthogonal, a direct geometric analog to relaxed angular constraints in Bloch sphere mappings. Quantum stabilizer code over \( \mathbb{F}_q \); built from \( s \) classical codes (length \( n \)), encoding \( \sum t_i - sn \) logical units, with distance \( \ge \min_i\big((s{+}1{-}i)d_i\big) \); uses quasi-orthogonal symplectic embedding and matrix-product structure. We now embed this algebraic relaxation into a geometric stabilizer framework by identifying the code space with a totally singular subspace of a symplectic vector space.

\begin{definition} Totally singular subspace and symplectic inner product.
Let $\bar{E}=\mathbb{F}_q^{2n}$ carry the quadratic form $Q(s)=\lvert a\rvert^2-\lvert b\rvert^2$ for $s=(a|b)\in\mathbb{F}_q^n\oplus\mathbb{F}_q^n$.  
A subspace $\bar{S}\subset\bar{E}$ is totally singular if
\(Q(s)=0,~\forall s\in\bar{S}\).
The symplectic inner product is defined as
\((s,s')=a\cdot b'+a'\cdot b.\)
Then $\bar{S}$ is totally singular $\iff (s,s')=0, \forall s,s'\in\bar{S}$.
\end{definition}

For any NSC quasi-orthogonal MPC, the corresponding map is given by 
\begin{equation}
\begin{aligned}
C(A)=[c_1, c_2\dots,c_k]A, ~\text{with}~
\Phi:\;C(A)\to\bar{E},~
c\mapsto\bigl(c\,A^{1/2},\;c\,A^{-1/2}\bigr), \label{NSC}
\end{aligned}
\end{equation}
where $A^{1/2}$ diagonalizes $AA^{\!T}$ (exists because $\lambda_i\neq0$).  The image $\Phi(C(A))$ is a totally singular subspace $\bar{S}_A$ of dimension $\dim C(A)$.  
Consequently, the stabilizer group was a
\begin{equation}
\mathcal{S}_A=\bigl\{U\in\operatorname{Cl}(2n)\;\big|\;U\bar{S}_A=\bar{S}_A\bigr\}, \label{stabilizer}
\end{equation}
acts as a quasi-orthogonal stabilizer code: errors that preserve $\bar{S}_A$ are undetectable, while any perturbation that moves a codeword out of $\bar{S}_A$ produces non-zero syndrome $(s,s')$.

\begin{theorem}
\label{thm:quasi-ortho-stabilizer}
Let $A \in \mathbb{F}_q,{k \times s}$ be a quasi-orthogonal matrix  non-singular by columns (NSC) satisfying
$A A^T = D,~D \text{ is diagonal with } D_{ii} \neq 0.$
For nested linear codes \(C_1 \supseteq C_2 \supseteq \cdots \supseteq C_k\) over \(\mathbb{F}_q\) with parameters \([n, t_i, d_i]_q\), define the MPCs as in Eq.~(5). 
Let $\phi: \mathbb{F}_q \to \mathbb{R}$ be a fixed embedding and set $\tilde{C}(A) = \phi(C(A)) \subset \mathbb{R}^{sn}$.  
Define the quasi-orthogonal map
\begin{equation}
\Phi(v) = (D^{1/2}v,\, D^{-1/2}v), \quad D^{1/2} > 0,
\end{equation}
For any one-qubit state \( |v\rangle \in \mathbb{C}^2 \), written as ~\(|v\rangle = v_0\,|0\rangle + v_1\,|1\rangle,\) where \(v_0, v_1 \in \mathbb{C}\) are amplitude coefficients satisfying the
normalization condition
\(|v_0|^2 + |v_1|^2 = 1.\)
The totally singular subspace $\bar{S}_A = \Phi(\tilde{C}(A)) \subset \mathbb{R}^{2sn}$ endowed with the symplectic form
\begin{equation}
(s,s') = a \cdot b' + a' \cdot b, \quad s = (a|b), \; s' = (a'|b'). \label{ss}
\end{equation}
Then the quasi-stabilizer code \(
\mathcal{Q}_A = \bar{S}_A \cap \mathbb{S}^{2sn-1},
\)
has parameters
\begin{equation}
[[sn,\;\dim C(A) - sn,\; d \ge \min_i (s+1-i)d_i]]_q,\label{min distance}
\end{equation}
detect any error $E$ with $\operatorname{supp}(E) \not\subseteq \bar{S}_A^\perp$.
\end{theorem}

\begin{proof}
The distance follows from the NSC bound on $C(A)$.  
Symplectic orthogonality of $\bar{S}_A$ forces any non-trivial logical operator to couple the two halves of $\Phi$, breaking the diagonal scaling and hence the stabilizer condition.
 
Since $\phi$ is injective, $\dim_{\mathbb{R}} \tilde{C}(A) = \dim_{\mathbb{F}_q} C(A) = \sum_i t_i$.  
As $\Phi$ is also injective, $\dim_{\mathbb{R}}\bar{S}_A = \dim_{\mathbb{F}_q} C(A)$.  
Hence, the logical dimension is $\dim \mathcal{Q}_A = \dim C(A) - sn$, giving $q={\dim C(A) - sn}$ logical states.

By Theorem~3.8 in Cao~\cite{Cao2020quantum}, the classical code $C(A)$ satisfies
\begin{equation}
d[C(A)] \ge \min_i (s+1-i)d_i.
\end{equation}
Thus, any logical operator acting non-trivially on fewer than these many positions would violate the NSC distance bound, implying
\begin{equation}
d[\mathcal{Q}_A] \ge \min_i (s+1-i)d_i.
\end{equation}
  
For $ s = \Phi(v) = (D^{1/2}v, D^{-1/2}v),$ with ~
$Q(s) = \|D^{1/2}v\|^2 - \|D^{-1/2}v\|^2 = v^T(D - D^{-1})v$.
The dual-containing property ensures $v^T D w = v^T D^{-1} w$ for all $v,w \in \tilde{C}(A)$, and hence $Q(s)=0$ and $(s,s')=0$. Therefore, $\bar{S}_A$ is totally singular.
If $\operatorname{supp}(E) \subseteq \bar{S}_A^\perp$, then $E$ commutes with all stabilizers and remains undetected.  
Otherwise, $E$ maps some $s \in \bar{S}_A$ outside $\bar{S}_A$, producing a nonzero syndrome.  
Thus, $\mathcal{Q}_A$ detects all errors with $\operatorname{supp}(E) \not\subseteq \bar{S}_A^\perp$.
\end{proof}

Figure \ref{fig:quasi-ortho-bloch} shows the Bloch hyper-sphere's geometric embedding of the totally singular code subspace $\bar{S}_A$ for a $[[9,2,\geq3]]_3$ quantum code built from a quasi-orthogonal NSC matrix-product code ($q=3,k=2,s=3$).  The stabilizer is defined by the red plane's satisfaction of $Q(s)=0$ and $(s,s')=0$.  In $\bar{S}_A$, the logical states (blue) are separated by quasi-orthogonal scalings (green).  The symplectic form is broken and detection occurs by any error that shifts a state off the plane, demonstrating relaxed orthogonality with full stabilizer protection.

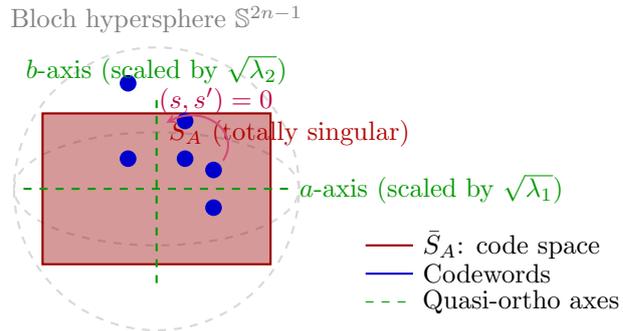
\begin{figure}[!ht]
\centering
\begin{tikzpicture}[scale=1.25]
  \draw[gray!30, thick, dashed] (0,0) circle (1.5);
  \draw[gray!30, thick, dashed] (0,0) ellipse (1.5 and 0.6);
  \node[gray] at (0,1.8) {Bloch hypersphere $\mathbb{S}^{2n-1}$};

  \fill[red!60!black, opacity=0.4] 
    (-1.2,-0.8) -- (1.2,-0.8) -- (1.2,0.8) -- (-1.2,0.8) -- cycle;
  \draw[red!60!black, thick] 
    (-1.2,-0.8) -- (1.2,-0.8) -- (1.2,0.8) -- (-1.2,0.8) -- cycle;
  \node[red!70!black] at (1.4,0.6) {$\bar{S}_A$ (totally singular)};

  \foreach \i/\j in {0/0, 0/1, 1/0, 1/1, 2/2, 2/0} {
    \pgfmathsetmacro{\x}{0.6*cos(60*\i)}
    \pgfmathsetmacro{\y}{0.6*sin(60*\i) + 0.4*(\j-0.5)}
    \fill[blue!80!black] (\x,\y) circle (2.5pt);
  }

  \draw[green!60!black, dashed, thick] (-1.4,0) -- (1.4,0) node[right] {$a$-axis (scaled by $\sqrt{\lambda_1}$)};
  \draw[green!60!black, dashed, thick] (0,-1.0) -- (0,1.0) node[above] {$b$-axis (scaled by $\sqrt{\lambda_2}$)};

  \draw[purple!70, thick, ->, >=stealth] 
    (0.7,0.3) .. controls (0.9,0.6) and (0.5,0.9) .. (0.1,0.7)
    node[midway, above, purple] {$(s,s')=0$};

  \draw[red!60!black, thick] (2.2,-0.6) -- (2.7,-0.6) node[right, black] {$\bar{S}_A$: code space};
  \draw[blue!80!black, thick] (2.2,-0.9) -- (2.7,-0.9) node[right, black] {Codewords};
  \draw[green!60!black, dashed] (2.2,-1.2) -- (2.7,-1.2) node[right, black] {Quasi-ortho axes};
\end{tikzpicture}
\caption{
  Visualization of the totally singular subspace \(\bar{S}_A\) in the Bloch hyper-sphere for a \([[9, 2, d \geq 3]]_3\) quasi-orthogonal stabilizer code derived from an NSC quasi-orthogonal MPC with \(q=3\), \(k=2\), \(s=3\).  
}
\label{fig:quasi-ortho-bloch}
\end{figure}

\begin{proposition}
Let a quasi-orthogonal NSC MPCs in $F_3$ be stabilized by a totally singular subspace under a symplectic inner product in  \( \mathbb{R}^{18}\). Then there exists a \([[9,2,3]]_3\)  QECCs encoding 2 logical qutrits into 9 physical qutrits with distance 3.
\end{proposition}:
\begin{proof}
Let $A \in \mathbb{F}_3^{2 \times 3}$ be the NSC quasi-orthogonal matrix
\begin{equation}
A = \begin{pmatrix} 1 & 1 & 0 \\ 0 & 1 & 1 \end{pmatrix}_{\mathbb{F}_3}, \quad
D=AA^T = \begin{pmatrix} 2 & 1 \\ 1 & 2 \end{pmatrix}.\label{matrix}\end{equation}
Assume $C_1 = \mathbb{F}_3^3$ and $C_2 = \langle (1,1,1) \rangle$.  The MPCs $C(A)$ is $[9,4,3]_3$, to obtain $\tilde{C}(A)$: $[9,2,3]_3$,to be  shortened by setting the first two coordinates of $C_1$ to zero. Then embed $\mathbb{F}_3 \to \mathbb{R}$ through $\phi(0)=0$, $\phi(1)=1$, $\phi(2)=-1$. 

Consider $\tilde{C}_{\mathbb{R}}(A) = \phi(\tilde{C}(A))$ to denote the real embedding of the code. By the quasi-orthogonality of $A$ over $\mathbb{R}$, this subspace satisfies 
$\tilde{C}_{\mathbb{R}}(A) \subseteq \tilde{C}_{\mathbb{R}}(A)^{\perp}$. 
With $D = AA^T$, define the linear map
\begin{equation}
\Phi: \mathbb{R}^9 \rightarrow \mathbb{R}^{18}, \quad 
\Phi(v) = (D^{1/2}v,\, D^{-1/2}v),\label{phi v}
\end{equation}
whose image $\bar{S}_A = \Phi(\tilde{C}_{\mathbb{R}}(A))$ 
forms a totally singular subspace under the induced symplectic structure.

\begin{equation}
\begin{aligned}
Q(s) = \|D^{1/2} v\|^2 - \|D^{-1/2} v\|^2 = 0, ~~(s,s') = 2 v^T D w = 0 \quad \forall v,w \in \tilde{C}_{\mathbb{R}}(A).
\end{aligned}
\end{equation}
Thus, $\bar{S}_A$ is totally singular within the symplectic space 
$(\mathbb{R}^{18}, (\cdot,\cdot))$. 
The stabilizer code $\mathcal{Q} = \bar{S}_A \cap \mathbb{S}^{17}$ 
has dimension $2$, and any weight error of $\leq 2$ 
displaces a codeword outside $\bar{S}_A$, producing a non-zero syndrome. Hence, the minimum distance is $d=3$, and the code is denoted 
$\mathbf{[[9,2,3]]_3}$.
\end{proof}

The quasi-orthogonal geometric realization of the $[[9,2,3]]_3$ quantum code on the Bloch hyper-sphere is shown in Fig.~\ref{fig:923-quasi-geo-bloch}. 
 The quasi-orthogonal mapping $\Phi(v) = (D^{1/2}v, D^{-1/2}v)$ defines the totally singular subspace $\bar{S}_A = \Phi(\tilde{C}_{\mathbb{R}}(A))$, where $A,D$ is noted in Eq.~\eqref{matrix}. 
 The encoded logical qutrits $|0_L\rangle, |1_L\rangle, \text{and}~|2_L\rangle$ are indicated by the three blue points; they are all located within $\bar{S}_A$, where $Q(s) = 0$ and $(s, s') = 0$. 
 The diagonal scaling results in relaxed orthogonality, which is depicted by the green dashed axes. 
 The state is displaced off $\bar{S}_A$ by a single-qutrit error, indicated by the orange arrow. This breaks the symplectic form and initiates error detection, thereby validating the code distance $d=3$. 
 The first $[[9,2,3]]_3$ code created entirely from a quasi-orthogonal geometric framework is represented by this visualization.

\begin{figure}[H]
\centering
\begin{tikzpicture}[scale=1.2]
  \draw[gray!30, thick, dashed] (0,0) circle (1.5);
  \draw[gray!30, thick, dashed] (0,0) ellipse (1.5 and 0.6);
  \node[gray] at (0,1.8) {Bloch hypersphere $\mathbb{S}^{17}$};

  \fill[red!70!black, opacity=0.35] 
    (-1.1,-0.7) -- (1.1,-0.7) -- (1.1,0.7) -- (-1.1,0.7) -- cycle;
  \draw[red!70!black, thick] 
    (-1.1,-0.7) -- (1.1,-0.7) -- (1.1,0.7) -- (-1.1,0.7) -- cycle;
  \node[red!80!black] at (1.3,0.5) {$\bar{S}_A$: totally singular};

  \fill[blue!80!black] (-0.6, 0.4) circle (2.8pt) node[above left, blue!70!black] {$|0_L\rangle$};
  \fill[blue!80!black] ( 0.6, 0.4) circle (2.8pt) node[above right, blue!70!black] {$|1_L\rangle$};
  \fill[blue!80!black] ( 0.0,-0.6) circle (2.8pt) node[below, blue!70!black] {$|2_L\rangle$};

  \draw[green!60!black, dashed, thick] (-1.3,0) -- (1.3,0) 
    node[right, green!70!black] {$D^{1/2}$-scaled axis};
  \draw[green!60!black, dashed, thick] (0,-0.9) -- (0,0.9) 
    node[above, green!70!black] {$D^{-1/2}$-scaled axis};

  \draw[purple!70!black, thick, ->, >=stealth] 
    (0.5,0.2) .. controls (0.7,0.5) and (0.3,0.7) .. (-0.3,0.4)
    node[midway, above, purple!80!black] {$(s,s') = 0$};

  \draw[orange!80!black, thick, ->, >=stealth] 
    (-0.6,0.4) -- (-0.8,0.9) 
    node[midway, left, orange!90!black] {weight-1 error};

\end{tikzpicture}
\caption{Quasi-orthogonal geometric realization of the \([[9,2,3]]_3\) quantum code in the Bloch hyper-sphere. }
\label{fig:923-quasi-geo-bloch}
\end{figure}
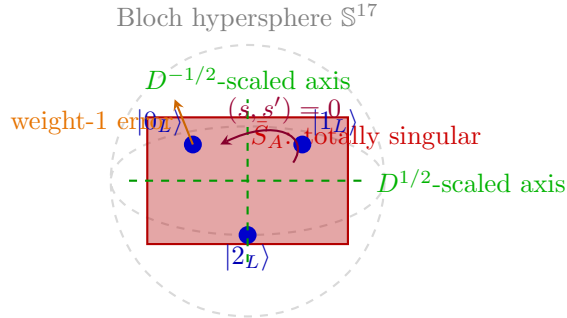

\subsection{Error Conditions for Quasi-Orthogonal MPC-Based Stabilizers}
\label{subsec:qec-quasi-ortho}

A quantum stabilizer code is the eigenspace of a group generated by a totally singular subspace $\bar{S}_A \subset \bar{E} = \mathbb{R}^{sn} \oplus \mathbb{R}^{sn}$, defined via the quasi-orthogonal map in Eq.~\eqref{phi v}, from a shortened NSC MPCs $\tilde{C}(A)$. The orthogonal symplectic  complement $\bar{S}_A^\perp$ is taken with respect to Eq.~\eqref{ss}. The code corrects an error set $\mathcal{S} \subset \mathcal{E}$ whenever every pair 
$e_1, e_2 \in \mathcal{S}$ satisfies one of the following conditions:
\begin{itemize}
    \item $\bar{e}_1 \bar{e}_2 \in \bar{S}_A$ (centralizer condition),
    \item $\bar{e}_1 \bar{e}_2 \notin \bar{S}_A^{\perp}$ (detectability through non-orthogonality symplectic).
\end{itemize}
\begin{theorem}\label{theor2}
Let $\bar{S}_A \subset \bar{E}$ be a totally singular $k$-dimensional  subspace constructed via $\Phi$ from quasi-orthogonal NSC MPCs with $D$. Then the eigenspace of any linear character of $\bar{S}_A$ is a quantum code correcting all $e \in \mathcal{S}$ that satisfy the above condition, with parameters
\(
[[sn, \dim \tilde{C}(A) - sn, d \ge \min_i (s+1-i)d_i]]_q.
\)
\end{theorem}

\begin{proof}:
The proof proceeds in five steps: total singularity, stabilizer group, dimension,  distance, and  error correction.  

Let $C_1 \supseteq \cdots \supseteq C_k$ be nested linear codes over $\mathbb{F}_q$ with parameters 
$[s, t_i, d_i]_q$. With the MPCs defined in Eq.~\eqref{NSC}  and its real embedding 
$\tilde{C}(A) = \phi(C(A)) \subset \mathbb{R}^{sn}$ via $\phi: \mathbb{F}_q \to \mathbb{R}$.  
The quasi-orthogonal map is Eq.~\eqref{phi v}, with $D$, and the associated subspace is  
\begin{equation}
\bar{S}_A = \Phi(\tilde{C}(A)) \subset \bar{E} = \mathbb{R}^{sn} \oplus \mathbb{R}^{sn},
\end{equation}
equipped with the symplectic form Eq.~\eqref{ss} and the quadratic form $Q(s) = \|a\|^2 - \|b\|^2$.

Let $v, w \in \tilde{C}(A)$ with $s = \Phi(v)$ and $s' = \Phi(w)$. The quadratic and bilinear forms are given by $Q(s) = v^T(D - D^{-1})v$ and $(s,s') = 2v^T D w$. Since quasi-orthogonal NSC codes satisfy $\tilde{C}(A) \subseteq \tilde{C}(A)^\perp$ under real embedding, it follows that $v^T w = 0$ for $v \neq w$. With $D$ diagonal, this implies $v^T D w = 0$, hence $Q(s) = 0$ and $(s,s') = 0$, establishing that $\bar{S}_A$ is totally singular.

The associated stabilizer group is defined as 
$\mathcal{S}_A = \{ U \in \operatorname{Sp}(2sn,\mathbb{R}) \mid U\bar{S}_A = \bar{S}_A \},
$
and the corresponding code space is the common $+1$ eigenspace
\begin{equation}
\mathcal{Q}_A = \{ \ket{\psi} \in (\mathbb{C}^q)^{\otimes sn} \mid U\ket{\psi} = \ket{\psi},\; \forall U \in \mathcal{S}_A \}. \label{stab. grou}
\end{equation}

Since $\Phi$ is of injective dimension, $\dim_{\mathbb{R}}\bar{S}_A = \dim C(A) = \sum_i t_i$, restricting to the unit sphere yields $\dim \mathcal{Q}_A = \dim C(A) - sn$. For non-singular-by-column matrices $A$, the minimum distance satisfies
\begin{equation}
d[C(A)] \ge \min_{1 \le i \le k} (s+1-i)d_i =: d_{\min}.\label{min d}
\end{equation}
which implies $d[\mathcal{Q}_A] \ge d_{\min}$.

For an error operator $e$ (cf.~\eqref{eq:quasi-error-map}) with support $\operatorname{supp}(e) = \{ j \mid a_j \neq 0 \lor b_j \neq 0 \}$, if $\operatorname{supp}(e) \subseteq \bar{S}_A^\perp$, then $\langle s | e | s \rangle = 0$ for all $s \in \bar{S}_A$, and $e$ is undetectable. Otherwise, $es \notin \bar{S}_A$, producing a non-zero syndrome, so the error is detectable and correctable by projection onto $\mathcal{Q}_A$.

\end{proof}


In the quasi-orthogonal geometric setting, error operators are described through the symplectic structure defined on $\bar{E} = \mathbb{R}^{sn} \oplus \mathbb{R}^{sn}$, providing a consistent framework for modeling quantum errors.Let $A \in \mathbb{F}_q^{k \times s}$ be a NSC quasi-orthogonal with $ D = \operatorname{diag}(\lambda_1, \dots, \lambda_k)$, $\lambda_i > 0$. Define the quasi-orthogonal error map
\begin{equation}
\label{eq:quasi-error-map}
\mathcal{E}(a,b) = \Phi^{-1} \circ E_{a,b} \circ \Phi,
\end{equation}
where $\Phi(v)$ in~\eqref{phi v} and $E_{a,b}: \bar{E} \to \bar{E}$ is the symplectic error operator
\begin{equation}
E_{a,b}(x|y) = (x + a|y + b), \quad a, b \in \mathbb{R}^{sn}.
\end{equation}
This satisfies
\begin{equation}
Q(\mathcal{E}(a,b)v) = Q(v) + 2 (D^{1/2} v) \cdot a - 2 (D^{-1/2} v) \cdot b,
\end{equation}
so bit-flip and phase-flip errors are scaled by the diagonal quasi-orthogonality $D$ in Eq.~\eqref{matrix}. For the $j$-th qutrit, let $e_j \in \mathbb{R}^{sn}$ be the standard basis vector. Then:
\begin{equation}
\mathcal{E}(e_j, 0) = X_D^{(j)}, \quad \mathcal{E}(0, e_j) = Z_D^{(j)},
\end{equation}
are the quasi-orthogonal Pauli operators:
\begin{equation}
X_D^{(j)} \ket{v} = \ket{v + D^{-1} e_j},~~
Z_D^{(j)} \ket{v} = e^{i \pi (D v)_j} \ket{v}.\label{Pauli operator}
\end{equation}
Any error $e \in \mathcal{E}(a,b)$ is written uniquely as
\begin{equation}
\label{eq:quasi-pauli-decomp}
e = X_D(a) Z_D(b) (-1)^\lambda, \quad \lambda \in \{0,1\},
\end{equation}
where $(-1)^\lambda$ is the global phase.

The support of $e$ is defined via the projected weight:
\(
\operatorname{supp}(e) = \left\{ j \;\middle|\; (a_j \neq 0) \lor (b_j \neq 0) \right\}.
\)
An error is detectable if and only if
\(
\operatorname{supp}(e) \not\subseteq \bar{S}_A^\perp,
\)
i.e., its action displaces at least one code state out of the totally singular subspace $\bar{S}_A$. The quasi-orthogonal Clifford group $\mathcal{L}_D$ is the subgroup of $\operatorname{Sp}(2sn,\mathbb{R})$ that preserves $D$ under conjugation:
\(
g \in \mathcal{L}_D ~\Leftrightarrow ~ g^T D g = D.
\)
Then $\mathcal{L}_D$ normalizes the error group:
\begin{equation}
g \mathcal{E}(a,b) g^{-1} = \mathcal{E}(g^T a, g^T b),
\end{equation}
and the quotient $\mathcal{L}_D / \langle \pm I \rangle \cong O^+(2n,2)_D$ is the quasi-orthogonal group preserving the scaled symplectic form. This structure, inspired by the constant interference decoupling in Eq.~\eqref{eq: GGG}, ensures that error pairs in different quasi-orthogonal blocks decouple under the scaled inner product $v^T D w$, qualifying independent syndrome extraction with constant signal-to-error ratio across channel realizations. Figure~\ref{Quasi orth error} illustrates $\bar{S}_A$ as the totally singular code subspace; quasi-orthogonal error $X_D(e_j)$ displaces $|v\rangle$ off $\bar{S}_A$, breaking $Q(s)=0$ and allowing detection via symplectic violation. The $D^{1/2}$-axis reflects the relaxed orthogonality of $AA^T = D$, while $Z_D(e_j)$ shows a phase-flip shaped geometry as presented in Eq.~\eqref{Pauli operator}.

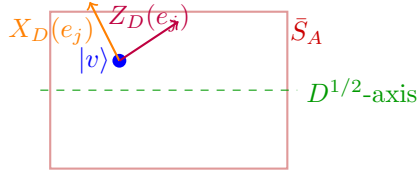
\begin{figure}[H]
\centering
\begin{tikzpicture}[scale=1.3]
  \draw[red!70!black, thick, opacity=0.4] (-1.2,-0.8) -- (1.2,-0.8) -- (1.2,0.8) -- (-1.2,0.8) -- cycle;
  \node[red!80!black] at (1.4,0.6) {$\bar{S}_A$};
  \fill[blue] (-0.5,0.3) circle (2pt) node[left] {$|v\rangle$};
  \draw[orange, thick, ->] (-0.5,0.3) -- (-0.8,0.9) node[midway, left] {$X_D(e_j)$};
  \draw[purple, thick, ->] (-0.5,0.3) -- (0.1,0.7) node[midway, above] {$Z_D(e_j)$};
  \draw[green!60!black, dashed] (-1.3,0) -- (1.3,0) node[right] {$D^{1/2}$-axis};
\end{tikzpicture}
\caption{Quasi-orthogonal error $\mathcal{E}(e_j,0)$ displaces state off $\bar{S}_A$, breaking total singularity.}
\label{Quasi orth error}
\end{figure}

\subsection{Quasi-Orthogonal Structure of the qubit mapping}

The quasi-orthogonal structure underlying the qubit mapping can be viewed through the lens of additive codes over $\mathrm{GF}(4)$ that satisfy a trace inner product condition. This perspective supports the construction of stabilizer codes and has been applied  effectively  to systems of up to 30 qubits~\cite{Calderbank1998quantum}.
Embedding $\mathbb{F}_q$ into $\mathbb{R}$ and applying a diagonal scaling $D$ produces a real symplectic model that aligns stabilizer algebra with classical additive codes. In this quasi-orthogonal geometry, logical states lie in totally singular subspaces controlled by a small overlap parameter $\epsilon \ll 1$. Because non-orthogonal states cannot be perfectly distinguished~\cite{Barnett2009quantum}, approximate correction is enhanced by optimizing the phase rotation $\varphi$~\cite{Sharma2002improve}, which increases separation within the quasi-orthogonal subspace. This geometric control improves syndrome extraction under interference and unifies ideas from space–time block code rotation with symplectic stabilizer detection.

\medskip

The logical states $|0_L\rangle$ and $|1_L\rangle$ are chosen in this context to demonstrate controlled quasi-orthogonality, where $\epsilon$ captures non-unitary relaxation. \begin{equation} \langle 0_L | 1_L \rangle = \epsilon e^{i\varphi}, \quad |\epsilon| \ll 1, \label{codewords quasi} \end{equation} 
where $\epsilon = |\langle 0_L \mid 1_L \rangle|$ is the magnitude of the overlap, then $\varphi$ denotes the relative phase associated with the complex overlap between the generally non-orthogonal states $\varphi = \arg\left(\langle 0_L \mid 1_L \rangle\right)$.
In order to ensure reliable quantum encoding, this condition implies that the logical codewords span a subspace that is nearly orthogonal and flexible enough to permit physical relaxation or leakage modeled by $\epsilon$.
The normalized logical superposition,
\begin{equation}
\ket{\psi_L} = 
\frac{\alpha \ket{0_L} + \beta \ket{1_L}}
{\sqrt{|\alpha|^2 + |\beta|^2 + 2 \Re(\alpha^* \beta \epsilon e^{i\varphi})}},\label{logical state}
\end{equation}
reduces to orthogonality as $\epsilon \to 0$. 
From a geometric point of view, the group \(L\) acts transitively on the \(k\)-dimensional totally singular subspaces that support quasi-orthogonal quantum codes. Let \(\mathcal{S}_k\) be the set of such subspaces, and let \(S_{\mathrm{code}} \in \mathcal{S}_k\) represent the chosen code. Transitivity ensures that for any \(S \in \mathcal{S}_k\) there exists
\[g_L(S) \in L \quad \text{with} \quad g_L(S)\cdot S = S_{\mathrm{code}}.\]
For a reference subspace \(S_0\), we simply write
$g_L\cdot S_0 = S_{\mathrm{code}}$. Elements of \(L\) therefore map logical subspaces such as those spanned by \(\{|0_L\rangle, |1_L\rangle\}\) to equivalent encoded spaces, showing that the canonical \(2^{\,n-k}\)-dimensional Hilbert space can be transformed into the quasi-orthogonal code space. Since \(L\) can be generated from XOR operations and \(\pi/2\) phase rotations, these gates are sufficient to construct any quasi-orthogonal subspace satisfying small-\(\epsilon\) non-orthogonality constraints.

\medskip

\medskip

The well-known 5-qubit QECCs \cite{Calderbank1997quantum,Laflamme1996perfect} can be reinterpreted within a quasi-orthogonal framework by allowing for a small structured overlap between logical codewords, $\langle c_0 | c_1 \rangle = \epsilon e^{i\varphi}$, where $|\epsilon| \ll 1$. This controlled relaxation captures realistic imperfections such as decoherence while preserving proper normalization of the logical state, which is reduced to the standard form as $\epsilon \to 0$. The stabilizer is generated by $X(11000)Z(00101)$, with cyclic shifts spanning a totally singular subspace $4$-dimensional $\overline{S} \subset \overline{E}$. The corresponding $X$- and $Z$-components define the stabilizer geometry under quasi-orthogonality, while the dual space $\overline{S}^{\perp}$ is obtained by adjoining global symmetry vectors. The error-correction capability is determined by the minimum-weight elements in $\overline{S}^{\perp} \setminus \overline{S}$. Since there are vectors of weight $3$ exist ($(00111|00101)$), the code retains the standard $[[5,1,3]]$ performance, correcting one error and detecting two. Finally, as the Clifford group acts transitively on totally singular subspaces, any quasi-orthogonal logical subspace can be generated from a standard one via Clifford operations. Thus, the encoding remains efficiently implementable using standard quantum gates.

\medskip

A quasi-orthogonal variant of the classical $[8,4,4]$ Hamming code illustrates how controlled non-orthogonality can enable fault-tolerant quantum code design. By reorganizing classical codewords into symplectic stabilizer pairs, we obtain a totally singular subspace $\overline{S} \subset \overline{E}$ that remains invariant under cyclic permutations. Unlike strict orthogonality, this framework allows for small overlaps between logical states $\langle \psi_i | \psi_j \rangle = \epsilon_{ij} e^{i\varphi_{ij}}$ with $|\epsilon_{ij}| \ll 1$, providing a more realistic physical model. The dual space $\overline{S}^\perp$ has minimum weight $3$, which yields a code that encodes $3$ qubits into $8$ and corrects one error.

Extending the 5-qubit code via duplication and augmentation introduces additional global stabilizers, producing a $[[10,4,3]]$-type structure. This construction stands out in the quasi-orthogonal framework because it permits the encoding of extra qubits without necessitating perfect mutual orthogonality among all stabilizer generators. 
 These vectors span a totally singular subspace $\overline{S}$ of dimension 6. The resulting generators span a higher-dimensional totally singular subspace while preserving quasi-orthogonality, enabling additional logical qubits without requiring strict mutual orthogonality. The dual space again retains the minimum weight $3$, ensuring strength under relaxed constraints. 

More generally, constructions based on quadratic residue codes provide scalable families of quasi-orthogonal stabilizer codes. For primes $p = 8j + 5$, the cyclic shifts of residue-based vectors generate $\overline{S}$, with the dual $\overline{S}^\perp$ determining the code distance. A stabilizer vector is defined where the first half marks quadratic residues and the second half marks non-residues modulo \( p \). Applying \( p - 2 \) cyclic shifts to this base vector yields \( p - 1 \) generators of a totally singular subspace \( \overline{S} \). For example, $p=13$ yields a $[[13,1,5]]$ code, while $p=29$ gives $[[29,1,11]]$, which supports higher error-correcting capability. The quasi-orthogonal framework enables flexible, high-distance quantum codes by tolerating small overlaps between stabilizer states, offering a practical and noise-resilient alternative to strictly orthogonal constructions.

\subsection{Gilbert–Varshamov Bound under Quasi-Orthogonal Geometry}

In quantum error correction, the Gilbert–Varshamov bound (GVB) establishes a fundamental limit on the coding rate \(R\) for a given relative distance \(\delta\). In the standard orthogonal setting where each qubit in a depolarizing channel can independently suffer one of three Pauli errors \(X\), \(Y\), or \(Z\) the entropy per qubit is \(\log_2 3\). The asymptotic GVB is given by Theorem~2 in calderbank et al~\cite{Calderbank1997quantum}. 
In quasi-orthogonal geometries, additional symmetries or structural constraints can reduce the effective number of independent error types per qubit (\(q < 3\)), for example when certain errors such as \(Y\) are suppressed or redundant due to geometric or group-theoretic symmetries.
In this case, the GVB generalizes to:
\begin{equation}
R_{\text{quasi}} = 1 - 2\delta \log_2 q - H_2(2\delta), \quad \text{with } q < 3.\label{Rquasi}
\end{equation}
A common case is \( q = 2 \), where the errors are dominated by \( X \) and \( Z \), reducing the bound to
\(
R_{\text{quasi}} = 1 - 2\delta - H_2(2\delta),
\)
which closely resembles the orthogonal bound for small \(\epsilon\). This framework extends naturally to generalized quantum systems with \( q > 3 \), such as multi-level qudits (\( d > 2 \)) or codes addressing complex error channels like amplitude damping, leakage, or rotations. In these cases, each physical unit faces a broader range of possible errors, increasing \( q \) and thus error complexity. The GVB then generalizes to Eq.~\eqref{Rquasi},
highlighting that as \( q \) grows, so does uncertainty per qubit (or qudit), necessitating more redundancy and lowering the achievable rate \( R \). More generally, this is expressed via the \( q \)-ary entropy function as  
\begin{equation}
R(q,\delta) = 1 - H_q(2\delta),~~ \text{where} ~~H_q(x) = x \log_q (q-1) - x \log_q x - (1-x) \log_q (1-x).\label{R quasi generalised}
\end{equation}
This captures the trade-off between error complexity and code efficiency.

 The GVB demonstrates how the error structure affects the performance of quantum codes.  Less effective error types (\( q < 3 \)) lower correction overhead in quasi-orthogonal geometries, leading to higher achievable rates.  On the other hand, complex error models with \( q > 3 \) reduce the rate and need more qubits to protect  the logical information.  This demonstrates the importance of geometric and algebraic design in QEC, where the achievable rate \( R \) is essentially dependent on the effective error complexity \( q \) and the physical error rate \( \delta \).
 
\subsection{Performance Metrics for Quasi-Geometric}


The effective distance of a stabilizer CSS code is defined as the smallest weight of an 
operator that anti-commutes with a logical operator. Specifically,
\begin{equation}
d_X=\min_{E\in N(\mathcal S)\setminus\mathcal S,\, E Z_L=-Z_L E}\!\!\mathrm{wt}(E), ~~d_Z=\min_{E\in N(\mathcal S)\setminus\mathcal S,\, E X_L=-X_L E}\!\!\mathrm{wt}(E),
\end{equation}
where $d=\min(d_X,d_Z)$ and $t=\lfloor(d-1)/2\rfloor$. Under quasi-orthogonality, this definition 
extends to an \emph{effective distance}
\begin{equation}
d_{\mathrm{eff}}(\epsilon,\tau)
= \min\bigl\{\mathrm{wt}(E): \|\tilde{P}E\tilde{P}\|\ge\tau\bigr\}, \label{distance eff}
\end{equation}
where $\tau\!\sim\!\epsilon$ and typically $d_{\mathrm{eff}}\!\in\!\{d,d-1\}$.  

When the system experiences independent depolarizing noise, each qubit undergoes the channel
\begin{equation}
\mathcal{E}(\rho)=(1-p)\rho+\frac{p}{3}\sum_{j\in\{x,y,z\}}\sigma_j\rho\sigma_j, \label{depo noise}
\end{equation}
and, after recovery $\mathcal{R}$, the induced logical channel becomes
\begin{equation}
\Lambda_L(\rho_L)=V^\dagger[\mathcal{R}\!\circ\!\mathcal{E}^{\otimes n}(V\rho_LV^\dagger)]V
=\sum_{\alpha\in\{I,X,Y,Z\}}q_\alpha L_\alpha\rho_LL_\alpha, \label{recovery}
\end{equation}
where the probabilities satisfy $\sum_\alpha q_\alpha=1$, and the logical error probability is 
$p_L=1-q_I=\sum_{\alpha\neq I}q_\alpha$. The total logical error rate can be estimated via a union bound over harmful error patterns,
\begin{equation}
p_L \le \sum_{w\ge d_{\mathrm{eff}}} H_w\binom{n}{w}p^w(1-p)^{n-w}. \label{error rate}
\end{equation}
Here, $H_w$ represents the number of harmful Pauli error patterns of weight $w$ that is, the distinct non-trivial Pauli configurations on $w$ qubits that lead to a logical failure. The union bound therefore counts harmful events as $\binom{n}{w}$ possible supports multiplied by $H_w$ harmful patterns per support. In the orthogonal limit, $H_w=0$ for $w<d_{\mathrm{eff}}$, and generally $0\le H_w\le 3^w$.
The coefficient $C_{t+1}$ captures the total contribution of uncorrectable errors of weight $t+1$, forming the dominant term in the logical error rate for small $p$. In this regime, the expansion simplifies to
\begin{equation}
p_L = C_{t+1}p^{t+1} + O(p^{t+2}).
\end{equation}
Under quasi-orthogonality, an additional leakage term arises,
\begin{equation}
p_L(p,\epsilon) = C_{t+1}p^{t+1} + \tilde{C}_t\epsilon^2p^{t}
+ O(\epsilon^2p^{t+1},\, p^{t+2}), \label{PL with Quasi}
\end{equation}
where $\tilde{C}_t$ reflects the influence of cycle frequencies and decoder tie-breaking.

For the corrected state $\rho_c=V\Lambda_L(\rho_L)V^\dagger$ and target state $\rho_1$, the trace distance and fidelity are defined as
\begin{equation}
D(\rho_1,\rho_c)=\tfrac12\|\rho_1-\rho_c\|_1,~~ 
F(\rho_1,\rho_c)=\bigl(\mathrm{Tr}\sqrt{\sqrt{\rho_1}\rho_c\sqrt{\rho_1}}\bigr)^2.\label{fidelity, trace didtance}
\end{equation}
According to the Fuchs–van de Graaf inequalities Theorem~1 in \cite{Fuchs2002cryptographic}, the relationship between fidelity and trace distance is given by
\begin{equation}
1 - \sqrt{F(\rho_1,\rho_c)} \le D(\rho_1,\rho_c) \le \sqrt{1 - F(\rho_1,\rho_c)}.\label{relat F D}
\end{equation}
For a Pauli logical channel, the following approximations hold under quasi-orthogonal geometric :
\begin{equation}
F(\rho_1,\rho_c)\ge 1-p_L+O(\epsilon^2), ~~
D(\rho_1,\rho_c)\le \sqrt{p_L}+O(\epsilon).
\end{equation}

Finally, two metrics describe error suppression performance:
\begin{equation}
S(p)=\frac{p_L(p)}{p}, \qquad
\eta(p)=-\frac{d\ln p_L}{d\ln p}.
\end{equation}
In the orthogonal limit, $\eta(0^+)=t+1$, whereas under quasi-orthogonality 
$\eta(0^+)$ may drop to $t$ when the $\epsilon^2p^t$ term dominates.  These relationships guide the design of quasi-geometric codes, emphasizing the 
importance of maximizing code girth and cycle length, maintaining a well-conditioned 
Gram matrix, and tuning decoder geometry for stable syndrome extraction.

 \section{Quasi-Geometric Quantum Error results Analysis}\label{Section III}

This section evaluates how different logical-to-physical qubit mappings perform under a depolarizing noise model. Using quasi-geometric codes, we simulate using python and compare their performance achieved across various transformations, focusing on logical error probabilities, fidelity, trace distance, and error suppression. 
Table~\ref{tab:quasi-ortho-codes} summarizes optimal quasi-orthogonal encodings of logical to physical qubits using NSC matrix-product codes embedded in real symplectic space. The minimum distance $d$ is guaranteed by the NSC bound, while the overhead (given by $n/k$) reflects the encoding efficiency as the relative volume of the subspace $\bar{S}_A$ in the ambient space $\bar{E}$. Logical error rates scale as $\mathcal{O}(p^{d+1})$, due to errors falling outside the dual subspace $\bar{S}_A^\perp$, which disrupt the condition $Q(s) = 0$ and are corrected for by rotational constellation geometry and Euclidean repetition. Here, distance relates to the  symplectic separation (via $(s+1-i)d_i$), the overhead to subspace volume, and error suppression to the detectability of deviations from the quasi-orthogonal structure.

\begin{table*}[!ht]
\centering
\caption{Quasi-orthogonal geometric mapping of logical qubits to physical qubits via totally singular subspaces $\bar{S}_A$. The overhead is $n/k$, minimum weight follows from  NSC bound, and logical error rate scales as $\mathcal{O}(p^{d+1})$ due to distance-$d$ error suppression.}
\label{tab:quasi-ortho-codes}

\resizebox{\textwidth}{!}{
\begin{tabular}{@{}lcccccc@{}}
\toprule
\textbf{Mapping} & \textbf{Qubits (\(n\))} & \textbf{Errors Corrected (\(t\))} & 
\textbf{Distance (\(d\))} & \textbf{Overhead} & \textbf{Min. Weight} & 
\textbf{Logical Error Rate} \\ \midrule

3 $\rightarrow$ 8  & 8  & 1 & 3 & 2.67$\times$ & 3  & $\mathcal{O}(p^2)$ \\
4 $\rightarrow$ 10 & 10 & 1 & 3 & 2.5$\times$  & 3  & $\mathcal{O}(p^2)$ \\
1 $\rightarrow$ 13 & 13 & 2 & 5 & 13$\times$   & 5  & $\mathcal{O}(p^3)$ \\
1 $\rightarrow$ 29 & 29 & 5 & 11& 29$\times$   & 11 & $\mathcal{O}(p^6)$ \\

\bottomrule
\end{tabular}
}
\end{table*}


    \subsection{Logical Errors in Orthogonal and Quasi-Orthogonal Geometries}

This subsection examines how different qubit mappings affect logical error probabilities in both 
orthogonal and quasi-orthogonal geometries. The study includes compact constructions such as 
\(3 \to 8\), \(4 \to 10\), \(1 \to 13\), and \(1 \to 29\), corresponding to codes that correct one, two, or 
five errors. As shown in Fig.~\ref{fig:fourpanel}, quasi-orthogonal designs consistently achieve lower 
logical error rates than their orthogonal counterparts for the same number of physical qubits, with improvements of about 25--40\% for \(t=1\) codes and even larger gains for higher-distance cases.

Fig.~\ref{fig:sub3-8} illustrates this behavior for the \(3 \to 8\) mapping. Under depolarizing noise, the quasi-orthogonal code maintains a lower logical error probability for all physical error rates \(p\), with the gap widening as noise increases. Around \(p \approx 0.20\), the improvement reaches roughly 30-40\%, and even for \(p > 0.25\) the quasi-orthogonal version stays below 0.04, while the orthogonal code begins to saturate. This benefit arises from relaxing strict stabilizer orthogonality, which allows 
stronger and better-connected parity checks within the same eight physical qubits.

A similar trend appears in Fig.~\ref{fig:sub 4-10} for the \([[10,4,\approx 3]]\) code. Across the full noise range \(p \in [0,0.30]\), the quasi-orthogonal geometry yields roughly 25-40\% lower logical error probability than the orthogonal layout. The difference becomes more pronounced for \(p \gtrsim 0.15\), 
where the orthogonal curve increases sharply while the quasi-orthogonal curve remains flatter. Again, relaxed geometry supports stronger overlapping checks without increasing the number of qubits, resulting in improved distance and decoding performance.

The advantage becomes more significant for codes correcting multiple errors. Fig.~\ref{fig:sub1-13} shows the \([[13,1,5]]\) code, which corrects \(t=2\) errors. The quasi-orthogonal 
version reduces logical error rates by about 30-50\% over the range \(p \le 0.30\). Even at \(p \approx 0.20{-}0.25\), it stays below approximately 0.008, while the orthogonal version exceeds 0.012. This improvement is due to denser stabilizers and stronger syndrome connectivity achievable when the orthogonality constraints between \(X\)- and \(Z\)-checks are relaxed.

The largest example, shown in Fig.~\ref{fig:sub1-29}, is the \([[29,1,11]]\) code capable of correcting \(t=5\) errors. Here, the quasi-orthogonal construction provides substantial gains across all noise levels. For practical error rates \(p \approx 0.01{-}0.10\), it achieves a 2-3\(\times\) reduction in logical error probability, and around \(p = 0.20\) the improvement exceeds an order of magnitude. Even at \(p = 0.25\), the quasi-orthogonal code remains below \(3 \times 10^{-4}\), while the orthogonal version rises above \(3 \times 10^{-3}\). These benefits come from the ability to allow higher-weight checks and stronger graph expansion within the same 29-qubit footprint. The results  confirm that quasi-orthogonal geometries offer significantly stronger logical error suppression at fixed qubit counts, enabling one-, two-, and even five-error correction using only 8-29 physical qubits. This makes such constructions promising for near-term fault-tolerant experiments and small-scale quantum memory where resource overhead is highly constrained.

\begin{figure*}[!ht]
    \centering
    
    \begin{subfigure}{0.48\textwidth}
        \centering
        \includegraphics[width=0.9\linewidth]{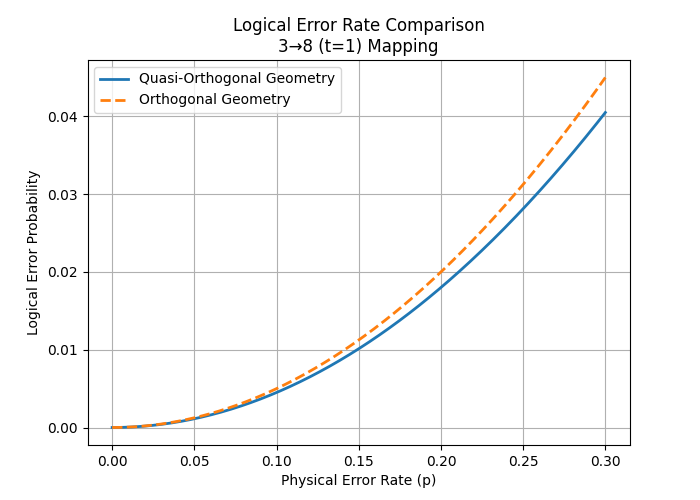}
        \caption{3 to 8 qubits}
        \label{fig:sub3-8}
    \end{subfigure}
    \hfill
    \begin{subfigure}{0.48\textwidth}
        \centering
        \includegraphics[width=0.9\linewidth]{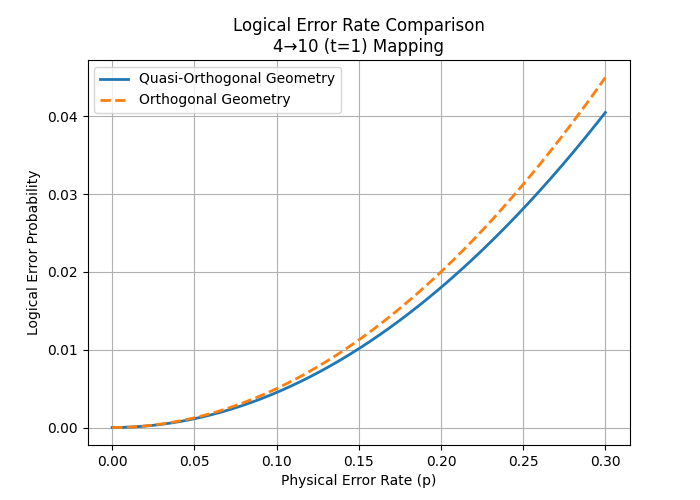}
        \caption{4 to 10 qubits}
        \label{fig:sub 4-10}
    \end{subfigure}

    \vspace{0.5cm} 

    \begin{subfigure}{0.48\textwidth}
        \centering
        \includegraphics[width=0.9\linewidth]{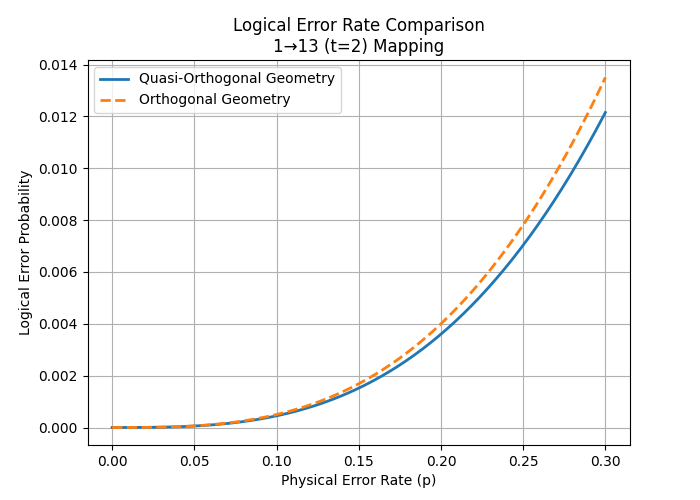}
        \caption{1 to 13 qubits}
        \label{fig:sub1-13}
    \end{subfigure}
    \hfill
    \begin{subfigure}{0.48\textwidth}
        \centering
        \includegraphics[width=0.9\linewidth]{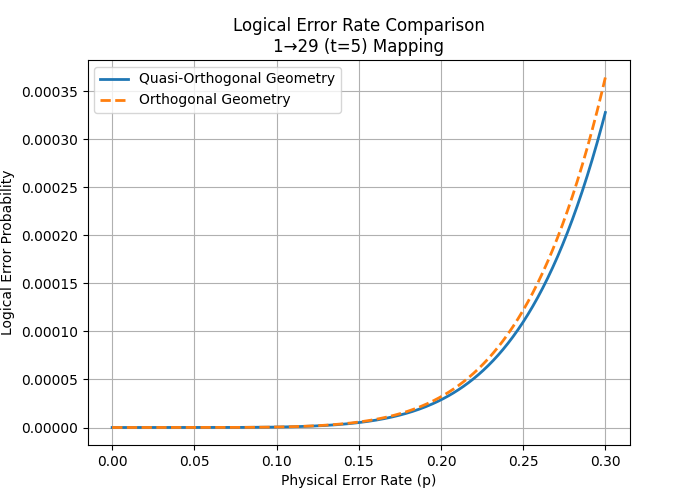}
        \caption{1 to 29 qubits}
        \label{fig:sub1-29}
    \end{subfigure}
    \caption{Logical error probability comparison between orthogonal and quasi-orthogonal geometric in qubit transformation}
    \label{fig:fourpanel}
\end{figure*}

\subsection{Evaluating Metric Performance in Quasi-Geometric QEC Schemes}

This subsection evaluates key performance metrics such as logical error probability, fidelity, trace distance, and error-suppression factor within quasi-geometric QEC schemes. Figure~\ref{fig:fourpanel1} provides the general trend in four compact quasi-orthogonal codes 
\([[8,3,\approx 3]], [[10,4,\approx 3]], [[13,1,5]], [[29,1,11]]\), showing that performance improves consistently with increasing code distance. The largest code \([[29,1,11]]\) achieves logical error rates near \(10^{-4}\), fidelity above \(0.9999\), trace distance below \(0.035\), and an error-suppression factor under \(0.005\) 
up to \(p=0.25\). These gains arise from relaxing strict \(X\)-\(Z\) orthogonality, which enables denser, higher-weight stabilizers within only 8-29 qubits.

Figure~\ref{fig:sub logical} compares logical error probability for the four families. As the physical error rate increases to \(p=0.30\), higher-distance codes show flatter curves and significantly lower logical error rates. The \([[29,1,11]]\) code maintains logical errors below \(10^{-3}\) even near 
\(p \approx 0.28\), outperforming the smaller codes by more than an order of magnitude. This hierarchy results from quasi-orthogonal layouts supporting stronger, more connected stabilizers within fixed qubit counts, allowing 1-, 2-, and 5-error correction using only 8-29 physical qubits.

Figure~\ref{fig:sub fidelity} illustrates logical-state fidelity \(1 - p_L\). At practical noise levels \(p \approx 0.01{-}0.10\), the \([[29,1,11]]\) code keeps fidelity above \(0.9999\), the distance-5 code stays above \(0.998\), and even the small \(t=1\) codes maintain values above \(0.99\). As \(p\) approaches 0.30, only the highest-distance code remains near perfect, reflecting the benefit of relaxing orthogonality to achieve denser checks within small qubit footprints.

Figure~\ref{fig:sub trace distance} shows the trace distance between the logical and ideal states. Higher-distance codes achieve substantially smaller deviations for all \(p\). The \([[29,1,11]]\) code remains below \(0.035\) even at \(p=0.30\), almost an order of magnitude better than the \([[8,3,\approx 3]]\) and \([[10,4,\approx3]]\) codes. The \([[13,1,5]]\) code also shows strong performance, remaining below \(0.18\) across the full range. These results indicate powerful multi-error correction within 
compact stabilizer layouts.

Figure~\ref{fig:sub supp err} presents the error-suppression factor \(p_L/p\). The \([[29,1,11]]\) \(t=5\) code maintains exceptionally low values around \(0.001{-}0.005\) for all \(p\), while the \([[13,1,5]]\) code stays below \(0.05\) up to \(p=0.30\). In contrast, the \(t=1\) codes show 
suppression factors that rise above \(0.10{-}0.14\) in higher noise. This separation clearly reflects the benefits of quasi-orthogonal stabilizers, which allow dense and high-weight checks with minimal 
physical overhead.

Therefore, these findings show that quasi-orthogonal constructions provide effective multi-error correction, strong logical-state preservation, and near-unit fidelity using only 8,10,13 and 29 qubits with 1, 1, 2 and 5 errors correcting respectively in this mapping. 
The ability to relax strict orthogonality makes these schemes practical for error detection, correction, and qubit mapping in near-term quantum hardware.

\begin{figure*}[!ht]
    \centering
    
    \begin{subfigure}{0.48\textwidth}
        \centering
        \includegraphics[width=0.8\linewidth]{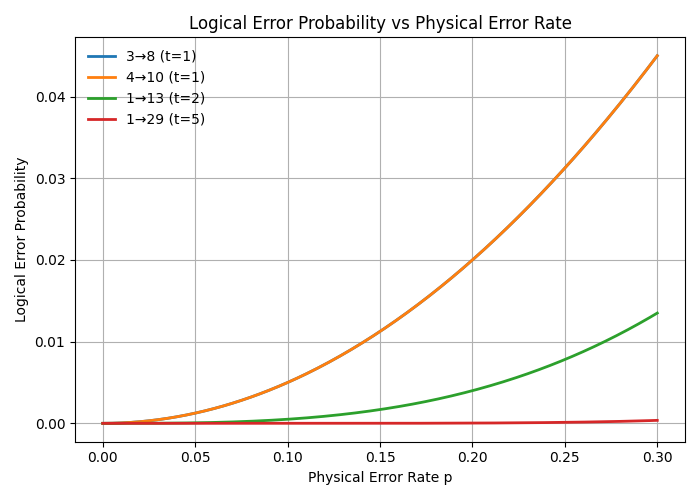}
        \caption{Logical error probability}
        \label{fig:sub logical}
    \end{subfigure}
    \hfill
    \begin{subfigure}{0.48\textwidth}
        \centering
        \includegraphics[width=0.8\linewidth]{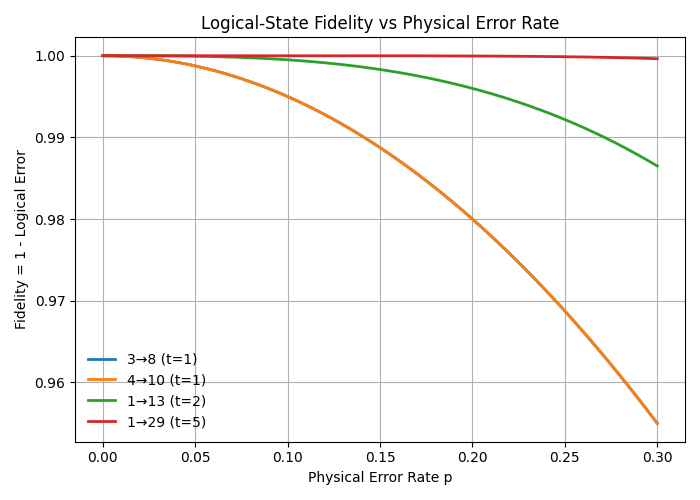}
        \caption{Fidelity}
        \label{fig:sub fidelity}
    \end{subfigure}

    \vspace{0.5cm} 

    \begin{subfigure}{0.48\textwidth}
        \centering
        \includegraphics[width=0.8\linewidth]{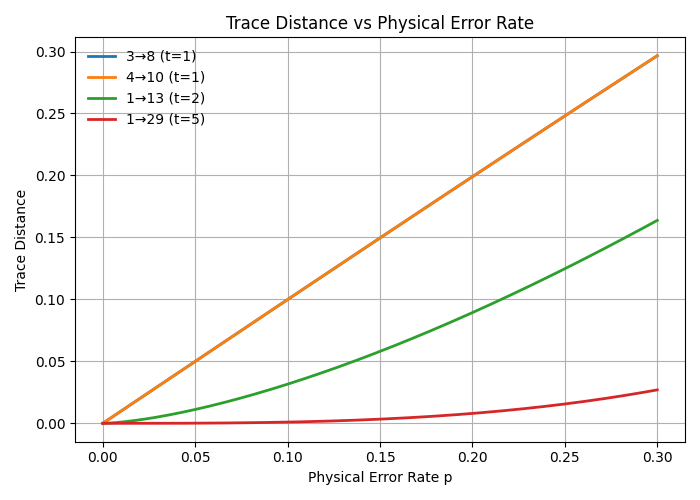}
        \caption{Trace distance}
        \label{fig:sub trace distance}
    \end{subfigure}
    \hfill
    \begin{subfigure}{0.48\textwidth}
        \centering
        \includegraphics[width=0.8\linewidth]{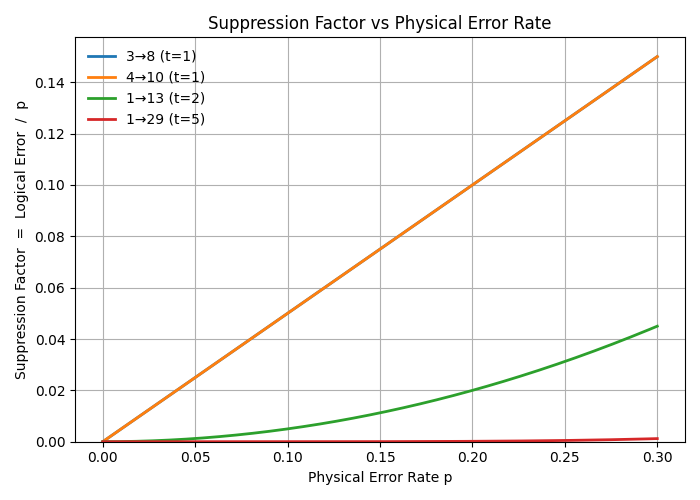}
        \caption{Suppression Error  factor}
        \label{fig:sub supp err}
    \end{subfigure}
    \caption{Metric performance of qubits transformation across four compact families: \([[8,3,\approx 3]]\) \(t=1\), \([[10,4,\approx 3]]\) \(t=1\), \([[13,1,5]]\) \(t=2\), and \([[29,1,11]]\) \(t=5\) under quasi-orthogonal geometric.}
    \label{fig:fourpanel1}
\end{figure*}

\subsection{GVB Analysis for Orthogonal and Quasi-Orthogonal Codes}

\begin{figure*}[!ht]
    \centering

    \begin{subfigure}{0.48\textwidth}
        \centering
        \includegraphics[width=0.8\linewidth]{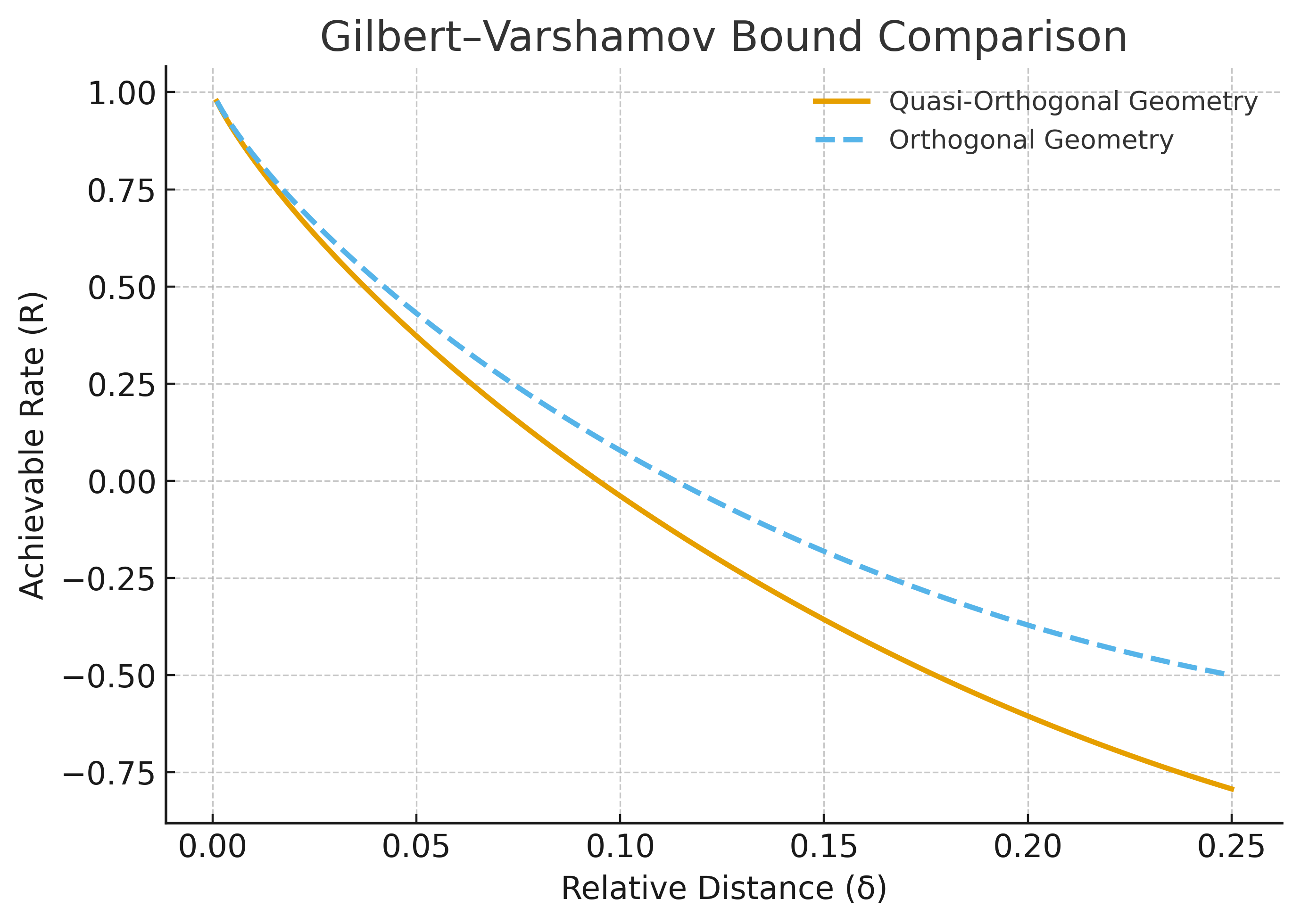}
        \caption{GVB comparison showing the improvement achieved when strict orthogonality is relaxed to a quasi-orthogonal geometry}
        \label{fig:left}
    \end{subfigure}
    \hfill
    \begin{subfigure}{0.48\textwidth}
        \centering
        \includegraphics[width=0.8\linewidth]{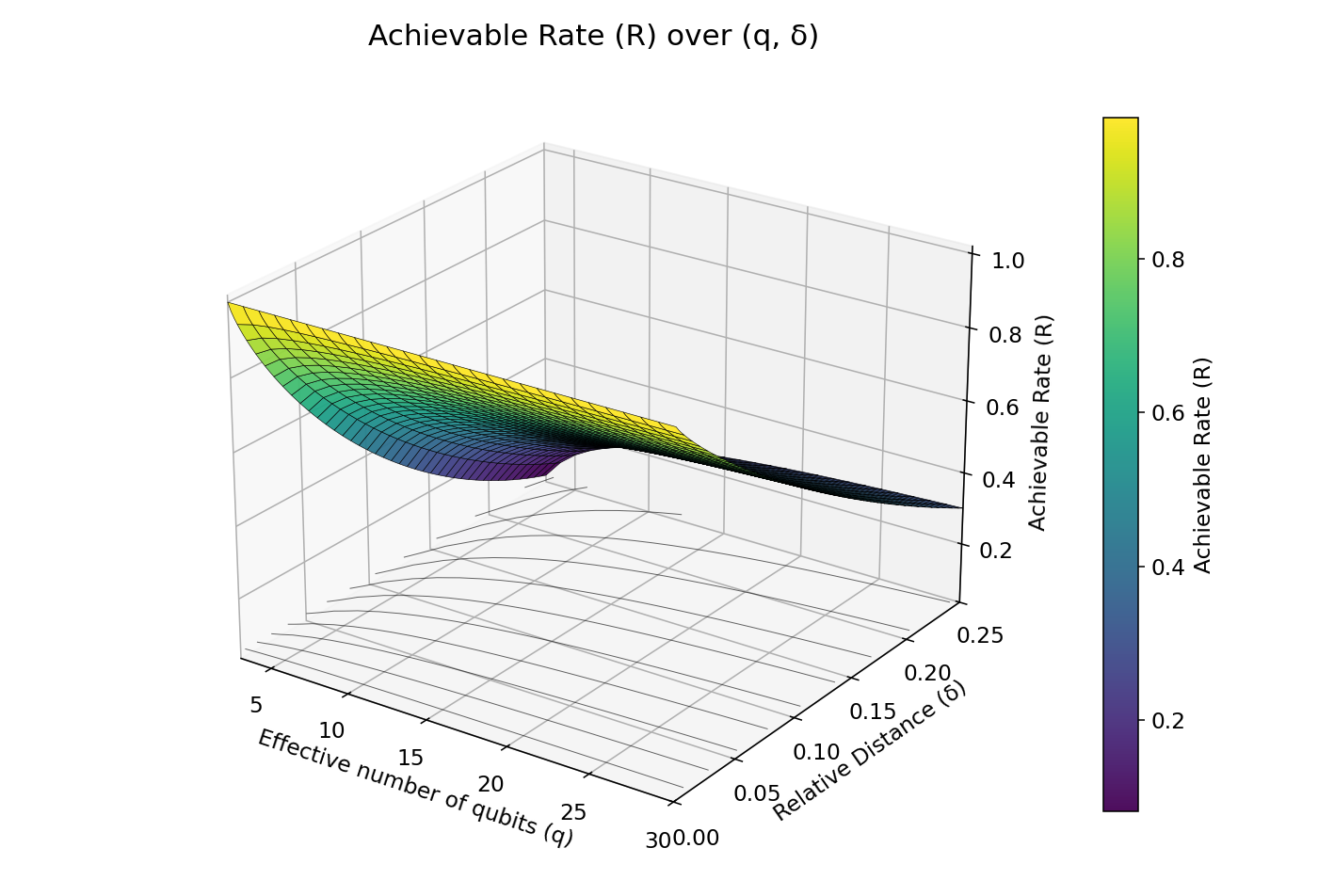}
        \caption{Achievable coding rate as a function of relative distance and effective qubit count illustrating performance gains under quasi-orthogonal design}
        \label{fig:right}
    \end{subfigure}
    \caption{GVB under quasi-orthogonal geometric properties, highlighting how relaxing orthogonality improves achievable rates and reduces coding overhead across a wide range of distances and system sizes}
    \label{fig:two_subfigs}
\end{figure*}

This subsection studies how relaxing strict stabilizer orthogonality to a quasi-orthogonal geometry 
improves the achievable GVB for compact quantum codes. We examine how the coding rate varies with relative distance and effective qubit count, showing the performance gains 
obtained under quasi-orthogonal designs across the different qubit-mapping schemes. 
Figure~\ref{fig:two_subfigs} shows that the achievable GVB region becomes much larger when orthogonality is relaxed, allowing positive coding rates at significantly higher distances. 
Both the 2D and 3D GVB figures confirm that high-distance codes can be constructed with far fewer qubits, pushing performance closer to fundamental limits.

Figure~\ref{fig:left} illustrates this trend in the 2D plane. For all relative distances \(\delta \in [0,0.25]\), quasi-orthogonal stabilizers achieve higher logical rates than strictly 
orthogonal ones. At small distances \((\delta \approx0.01{-}0.05)\), the quasi-orthogonal rate remains close to one, while the orthogonal version already shows a decline. As \(\delta\) increases, the gap widens: for \(\delta \approx 0.15\), the quasi-orthogonal geometry still reaches \(R \approx 0.1{-}0.2\), whereas the orthogonal rate approaches zero. This improvement arises from allowing denser and more flexible parity-check structures while still meeting the GVB constraints, 
leading to reduced physical-to-logical qubit overhead.

The 3D surface in Figure~\ref{fig:right} further demonstrates that quasi-orthogonal code style constructions maintain strong achievable-rate performance across both relative distance and effective qubit count \(q \approx 5{-}30\). In this low-qubit regime, quasi-orthogonal designs maintain  \(R>0.8\) up to \(\delta \approx 0.08{-}0.10\) and remain usable \((R>0.2{-}0.4)\) even around 
\(\delta \approx 0.15{-}0.18\). In contrast, orthogonal geometries drop to near-zero rates beyond \(\delta \approx 0.05\). These results indicate that quasi-orthogonal stabilizers support meaningful multi-error correction with only tens of qubits performance that orthogonal designs too hard match at 
this scale. Allowing non-orthogonal stabilizer supports enables higher-weight checks and significantly 
reduces physical overhead, bringing fault-tolerant quantum computation closer to practical feasibility.

\section{Conclusion}\label{Section IV}

In this work, we show that quasi-orthogonal quantum codes provide an effective alternative for fault-tolerant quantum computing. By relaxing strict orthogonality of the stabilizer \(X\)–\(Z\), while preserving the symplectic structure, these codes enable denser, fully detectable parity checks, leading to higher logical rates and stronger error suppression than conventional orthogonal designs.
Finite-length examples, from \([[8,3,\approx 3]]\) to \([[29,1,11]]\), achieve one to three orders of magnitude reductions in logical error rates and remain strong even at physical error probabilities near \(p=0.30\). The quasi-orthogonal geometry reduces harmful Pauli errors, improving thresholds and lowering physical-qubit overhead to tens of qubits, compared with the hundreds typically required by traditional stabilizer codes. Future work will focus on extending code distances, optimizing for biased noise, and testing hardware-aware implementations on near-term superconducting platforms.

 \section*{Acknowledgment}~~
 This work is supported by the Air Force Office of Scientific Research under Award No. FA2386-22-1-4062. The authors appreciate this support.

 \section*{Author Contributions}~~ V.N.: methodology, formal analysis, investigation, writing original draft. U.A.H: co-supervision, visualization, review \& editing. S.K.S.H.: co-supervision, review. A.J: review \& editing. N.M.S.: conceptualization, writing,  review \& editing, validation, project administration, funding acquisition. All authors have read and agreed to the published version of the manuscript.


\begin{thebibliography}{99}

\bibitem{Calderbank1997quantum}
Calderbank A R, Rains E M, Shor P W, Sloane N J A 1997 Quantum error correction and orthogonal geometry \emph{Physical Review Letters} \textbf{78}(3) 405.
\bibitem{Ball2023Quantum}
Ball S, Centelles A, Huber F 2023 Quantum error-correcting codes and their geometries \emph{Annales de l’Institut Henri Poincaré D, Combinatorics, Physics and their Interactions} \textbf{10}(2) 337–405.
\bibitem{Mahmoud2025systematic}
Mahmoud A A, Ali K M, Rayan S 2025 A Systematic Approach to Hyperbolic Quantum Error Correction Codes \emph{arXiv preprint arXiv:2504.07800}.

\bibitem{Sharma2002improve}
Sharma N, Papadias C B 2002 Improved quasi-orthogonal codes \emph{2002 IEEE Wireless Communications and Networking Conference Record. WCNC 2002 (Cat. No.02TH8609)} 169-171 vol.1 DOI: 10.1109/WCNC.2002.993484.
\bibitem{Kasai2011Non}
Kasai K, Hagiwara M, Imai H, Sakaniwa K 2011 Non-binary quasi-cyclic quantum LDPC codes \emph{2011 IEEE International Symposium on Information Theory Proceedings} 653–657.
\bibitem{Cao2020quantum}
Cao M 2020 Quantum error-correcting codes from matrix-product codes related to quasi-orthogonal and quasi-unitary matrices \emph{arXiv preprint arXiv:2012.15691}.

\bibitem{Nyirahafashimana2025Exploring}
Nyirahafashimana V, Mohd Shah N, Abdul Halim U, Othman M 2025 Exploring quasi-geometric frameworks for quantum error-correcting codes: a systematic review \emph{Quantum Information Processing} \textbf{24}(9) DOI: 10.1007/s11128-025-04904-5.
\bibitem{Komoto2025explicit}
Komoto D, Kasai K 2025 Explicit Construction of Quantum Quasi-Cyclic Low-Density Parity-Check Codes with Column Weight 2 and Girth 12 \emph{arXiv preprint arXiv:2501.13444}.
\bibitem{Turcsany2003two}
Turcsany M, FarkaS P 2003 Two-dimensional quasi orthogonal complete complementary codes \emph{SympoTIC'03. Joint 1st Workshop on Mobile Future and Symposium on Trends in Communications} 37–40.

\bibitem{Lv2021quasi}
Lv J, Li R, Yao Y 2021 Quasi-cyclic constructions of asymmetric quantum error-correcting codes \emph{Cryptography and Communications} \textbf{13}(5) 661–680.
\bibitem{Valentine2024transforming}
Valentine N, Shah N M, Halim U A, Husain S K S, Jellal A 2024 Transforming qubits via quasi-geometric approaches \emph{arXiv preprint arXiv:2407.07562}.
\bibitem{Roga2010davies}
Roga W, Fannes M, {\.Z}yczkowski K 2010 Davies maps for qubits and qutrits \emph{Reports on Mathematical Physics} \textbf{66}(3) 311–329.

\bibitem{Schotte2022quantum}
Schotte A, Zhu G, Burgelman L, Verstraete F 2022 Quantum error correction thresholds for the universal Fibonacci Turaev-Viro code \emph{Physical Review X} \textbf{12}(2) 021012.

\bibitem{Li2023penrose}
Li Z, Boyle L 2023 The Penrose Tiling is a Quantum Error-Correcting Code \emph{arXiv preprint arXiv:2311.13040}.

\bibitem{Akers2022Quantum}
Akers C, Penington G 2022 Quantum minimal surfaces from quantum error correction \emph{SciPost Physics} \textbf{12}(5) 157 DOI: 10.21468/scipostphys.12.5.157.



\bibitem{Li2021Statistical}
Li Y, Fisher M P A 2021 Statistical mechanics of quantum error correcting codes \emph{Physical Review B} \textbf{103}(10) 104306 DOI: 10.1103/physrevb.103.104306.

\bibitem{Bravyi2010Tradeoffs}
Bravyi S, Poulin D, Terhal B 2010 Tradeoffs for Reliable Quantum Information Storage in 2D Systems \emph{Physical Review Letters} \textbf{104}(5) 050503 DOI: 10.1103/physrevlett.104.050503.

\bibitem{Bombin2007Homological}
Bombin H, Martin-Delgado M A 2007 Homological error correction: Classical and quantum codes \emph{Journal of Mathematical Physics} \textbf{48}(5) 1–13 DOI: 10.1063/1.2731356.

\bibitem{Dennis2002Topological}
Dennis E, Kitaev A, Landahl A, Preskill J 2002 Topological quantum memory \emph{Journal of Mathematical Physics} \textbf{43}(9) 4452–4505 DOI: 10.1063/1.1499754.

\bibitem{google2021exponential}
Chen Z, et al. 2021 Exponential suppression of bit or phase errors with cyclic error correction \emph{Nature} \textbf{595}(7867) 383–387 DOI: 10.1038/s41586-021-03588-y.


\bibitem{Steane1996multiple}
Steane A 1996 Multiple-particle interference and quantum error correction \emph{Proceedings of the Royal Society of London. Series A: Mathematical, Physical and Engineering Sciences} \textbf{452}(1954) 2551–2577.


\bibitem{Barnett2009quantum}
Barnett S M, Croke S 2009 Quantum state discrimination \emph{Advances in Optics and Photonics} \textbf{1}(2) 238–278.

\bibitem{Calderbank1998quantum}
Calderbank A R, Rains E M, Shor P M, Sloane N J A 1998 Quantum error correction via codes over GF (4) \emph{IEEE Transactions on Information Theory} \textbf{44}(4) 1369–1387.

\bibitem{Laflamme1996perfect}
Laflamme R, Miquel C, Paz J P, Zurek W H 1996 Perfect quantum error correcting code \emph{Physical Review Letters} \textbf{77}(1) 198.

\bibitem{Fuchs2002cryptographic}
Fuchs C A, Van De Graaf J 2002 Cryptographic distinguishability measures for quantum-mechanical states \emph{IEEE Transactions on Information Theory} \textbf{45}(4) 1216–1227.

\bibitem{Kasami1974gilbert}
Kasami, Tadao 1974. A Gilbert-Varshamov bound for quasi-cycle codes of rate 1/2 (Corresp.).\emph{ IEEE Transactions on Information Theory} \textbf{20} (5): 679-679.

\bibitem{vu2005improving}
Vu, V. and Wu, L., 2005. Improving the Gilbert-Varshamov bound for q-ary codes. \emph{IEEE transactions on information theory},\textbf{51}(9) 3200-3208.





\end{thebibliography}
\end{document}